\documentclass[letterpaper, 10 pt, conference]{ieeeconf}

\IEEEoverridecommandlockouts

\usepackage{amsmath}
\usepackage{color}
\usepackage{algorithm}
\usepackage{tikz}
\usepackage{graphicx}
\usepackage{authblk}
\usepackage{mathtools}

\newtheorem{define}{Definition}

\newtheorem{problem}{Problem}
\newtheorem{lemma}{Lemma}
\newtheorem{remark}{Remark}
\newtheorem{assumption}{Assumption}
\newtheorem{theorem}{Theorem}

\newtheorem{proposition}{Proposition}

\title{\LARGE \bf  Safety Control of Positive Monotone Systems with Bounded Uncertainties}

\author{Sadra Sadraddini and Calin Belta 
\thanks{The authors are with the Department of Mechanical Engineering, Boston University, Boston, MA 02215  \{sadra,cbelta\}@bu.edu. \\
This work was partially supported by the NSF under grants CPS- 1446151 and CMMI-1400167.}
}

\usepackage{amsfonts}

\begin{document}

\maketitle

\thispagestyle{empty}
\pagestyle{empty}

\begin{abstract}
Monotone systems are prevalent in models of engineering applications such as transportation and biological networks. In this paper, we investigate the problem of finding a control strategy for a discrete time positive monotone system with bounded uncertainties such that the evolution of the system is guaranteed to be confined to a safe set in the state space for all times. By exploiting monotonicity, we propose an approach to this problem which is based on constraint programming. We find  control strategies that are based on repetitions of  finite sequences of control actions. We show that, under assumptions made in the paper, safety control of cooperative systems does not require state measurement. We demonstrate the results on a signalized urban traffic network, where the safety objective is to keep the traffic flow free of congestion.
\end{abstract}

\section{Introduction}
Designing control policies subject to safety constraints is a fundamental problem in the automation of complex systems. 
From a game theoretic perspective, the safety control problem, also known as safety game, is the problem of finding a  control policy  that guarantees that the evolution of the system is restricted to a safe region in the state space, regardless of the actions taken by the adversary. The solution to this problem involves finding a \emph{robust control invariant set} \cite{blanchini1999survey}.
Iterative computation of robust control invariant sets has been extensively studied  for linear and piecewise affine systems \cite{kerrigan2001}\cite{rakovic2004computation}, where intensive polyhedral operations are required to carry out set iterations. 

In this work, we focus on a special class of systems that are monotone, or order preserving, and provide an alternative computational approach to the safety control problem.
cooperative systems are common in models of biological, socio-economical and transportation networks. Monotonicity, in general, is a mathematical property that indicates a type of order preserving law. Monotone autonomous systems are thoroughly studied in \cite{smith2008monotone}. In \cite{angeli2003monotone}, the authors introduced cooperative control systems and provided results on steady state responses and stability.

We consider discrete time uncertain control systems that are monotone with respect to positive orthant in the state and adversarial inputs space. In contrast to \cite{angeli2003monotone}, we do not assume monotonicity with respect to controls. We do not even require the control space to be partially ordered. On the other hand, we assume a more restrictive form of the safety region in the problem formulation. 
Our consideration of such systems and specifications is motivated by the dynamics of urban traffic networks \cite{coogan2015controlling}, which are described in more detail later in the paper. The key result of this work is to show that computing robust control invariant sets maps to computing finite sequences of control actions, which we call \emph{s-sequences}. We show that repeated executions of s-sequences are safe control policies that do not require state feedback. We also show that, under some mild assumptions, the existence of s-sequences is almost necessary. To the best of our knowledge, these fundamental insights were not established before.

Safety control of monotone systems has also been considered in \cite{ghaemi2011safety} and \cite{meyer2015safety}. However, in these papers, monotonicity with respect to the controls was also assumed. Therefore, the results of this paper are more general in this respect. Set-invariance theories are also closely related to stability analysis. In \cite{smith2008monotone}, \cite{forghani2015safety}, \cite{forghani2016design}, \cite{Coogan2014}, \cite{lovisari2014stability}, the authors studied the stability of monotone and mixed monotone deterministic systems with no control inputs.  Extending these results to cooperative systems with partially ordered adversarial inputs is relatively straightforward, but it is not so obvious for systems with control inputs, specifically for discontinuous  control admissible sets.
 
This work is also related to finite state abstraction based control of (mixed) monotone systems \cite{coogan2015efficient}. This approach enables control synthesis from rich temporal logic \cite{baier2008principles} specifications, of which safety is a special yet important class. However, discretization of the state space is computationally expensive and its complexity grows exponentially with respect to the size of the system. Furthermore, with particular focus on safety specifications of the form assumed in this paper, our results are stronger in the following ways. First, if our approach does not find a solution to the safety control problem, we are almost certain that a solution by any approach does not exist. This result is rarely achieved in finite state abstraction based control, unless a bisimulation quotient is constructed (see, e.g, \cite{tabuada2009verification}). Second, we find policies that do not require feedback, hence implementing the control loop does not require sensing. Third, our method is computationally more efficient.

This paper is organized as follows. We provide the necessary notation in Sec. \ref{sec:prelim} and  formulate the problem in Sec. \ref{sec:problem}. In Sec. \ref{sec:robust}, we show how to compute  robust control invariant sets and s-sequences. In Sec. \ref{sec:cycle}, we characterize the long term response of the system to repeated s-sequences.  In Sec. \ref{sec:necessary}, we explain the underlying assumptions and formalize the notion of almost necessity for the existence of s-sequences. Finally, we provide two case studies in Sec. \ref{sec:case}.

%
%
%

\section{Preliminaries}
\label{sec:prelim}

We denote the positive orthant of an $n$-dimensional space by $\mathbb{R}_+^n:=[0,\infty)^n$. For two vectors $a,b \in \mathbb{R}^n$, we use the following notations:
\begin{equation}
\begin{array}{ccc}
a \prec b & \Leftrightarrow & a_i < b_i, \\
a \preceq b & \Leftrightarrow & a_i \le b_i, \\
\end{array}
\end{equation}
for all $i=1,\cdots,n$. We denote the $n$-dimensional vector of all ones by $1_n$.
\begin{define}
Given a vector $a \in \mathbb{R}_+^n$, the set $\mathcal{R}(a)$ is defined as:
\begin{equation}
\mathcal{R}(a):=\left \{ x \in \mathbb{R}_+^n \left | \right. x \preceq a \right \}.
\end{equation} 
\end{define}
\begin{define}\cite{EricS.KimMuratArcak2016}
The set $\mathcal{S} \subseteq \mathbb{R}_+^n $ is a \emph{lower-set} if $\forall x \in \mathcal{S}$ we have $\mathcal{R}(x) \subseteq \mathcal{S}$.
\end{define}

\begin{figure}[t]
\centering
\begin{tikzpicture}[xscale=0.7,yscale=0.7]
\draw[fill=green!50]  (0,0) -- (3,0) -- (4,0) -- (4,1) -- (3,1) --  (2,3) -- (1,3.6) -- (1,4.4) -- (0,4.4) -- (0,3.6) -- cycle;
\draw[dashed,color=black!50]  ( 0.6 ,0) -- ( 0.6 ,5);
\draw[dashed,color=black!50]  (0, 0.4 ) -- (5, 0.4 );
\draw[dashed,color=black!50]  ( 1.2 ,0) -- ( 1.2 ,5);
\draw[dashed,color=black!50]  (0, 0.8 ) -- (5, 0.8 );
\draw[dashed,color=black!50]  ( 1.8 ,0) -- ( 1.8 ,5);
\draw[dashed,color=black!50]  (0, 1.2 ) -- (5, 1.2 );
\draw[dashed,color=black!50]  ( 2.4 ,0) -- ( 2.4 ,5);
\draw[dashed,color=black!50]  (0, 1.6 ) -- (5, 1.6 );
\draw[dashed,color=black!50]  ( 3.0 ,0) -- ( 3.0 ,5);
\draw[dashed,color=black!50]  (0, 2.0 ) -- (5, 2.0 );
\draw[dashed,color=black!50]  ( 3.6 ,0) -- ( 3.6 ,5);
\draw[dashed,color=black!50]  (0, 2.4 ) -- (5, 2.4 );
\draw[dashed,color=black!50]  ( 4.2 ,0) -- ( 4.2 ,5);
\draw[dashed,color=black!50]  (0, 2.8 ) -- (5, 2.8 );
\draw[dashed,color=black!50]  ( 4.8 ,0) -- ( 4.8 ,5);
\draw[dashed,color=black!50]  (0, 3.2 ) -- (5, 3.2 );
\draw[dashed,color=black!50]  (0, 3.6 ) -- (5, 3.6 );
\draw[dashed,color=black!50]  (0, 4.0 ) -- (5, 4.0 );
\draw[dashed,color=black!50]  (0, 4.4 ) -- (5, 4.4 );
\draw[->] (0,0) -- (0,5);
\draw[->] (0,0) -- (5,0);
\draw (1.5,1) node[] {$\mathcal{S}$};
\end{tikzpicture}
\caption{A lower-set $\mathcal{S} \subset \mathbb{R}^2_+$.}
\label{fig:box}
\end{figure}

A graphical illustration of a lower-set is depicted in Figure \ref{fig:box}. Note that lower-sets can be non-convex. 
\begin{proposition}
\label{prop:closure}
The set of lower-sets is closed under union and intersection, i.e. if the sets $\mathcal{S}_1$ and $\mathcal{S}_2$ are lower-sets, then $\mathcal{S}_1 \cup \mathcal{S}_2$ and $\mathcal{S}_1 \cap \mathcal{S}_2$ are also lower-sets. 
\end{proposition}

\section{Problem Statement and Approach}
\label{sec:problem}

\subsection{Motivating Application: Urban Traffic Networks}

An urban traffic network is usually modeled as a directed graph, where its edges and vertices represent traffic links and junctions, respectively. An example of an urban traffic network is shown in Figure \ref{fig:network}. We adopt the discrete time fluid-like vehicular flow model from \cite{coogan2015controlling}, which is briefly explained in Sec. \ref{sec:case2}. The control input is the set of red/green light decisions at the junctions and the adversarial inputs are the numbers of exogenous vehicles arriving in each link in one time step. An upper bound for the adversarial input of each link is assumed to be known.  From a game theoretical view, the aim of the adversary is to congest the network, while the winning condition for the player is to keep the links free of congestion. 

Monotonicity in traffic networks indicates that given a fixed sequence of control actions, an increase in the vehicular occupancy of some link leads to subsequent higher or at least equal level of occupancy in the whole network at later times. However, traffic networks are not fully cooperative. It is shown in \cite{Coogan2014} that under a \emph{first in first out} (FIFO) rule, monotonicity does not hold at diverging junctions. For instance, consider the flow in  links $2,3,10$ in Figure \ref{fig:network}. If the number of vehicles on link $3$ is near its capacity, then it limits the vehicular flow from  link $2$. On the other hand, under FIFO policy, the flow of the vehicles from  link $2$ to $10$ is also impeded. Consequently, an increase in the occupancy of  link $3$ may actually decrease the occupancy of link $10$. The authors in \cite{coogan2015mixed} studied this phenomenon and showed that traffic networks are \emph{mixed monotone}, which is a weaker property than monotonicity.  

We desire that links do not impede the vehicular flow from their upstream links, i.e. the situation described above never happens. In other words, we desire the traffic network to behave as a cooperative system. The set of states that correspond to cooperative dynamics is called \emph{cooperative region}, which is straightforward to show that is a lower-set in the state space, i.e. it always favors less amount of vehicles. Therefore, it is practically meaningful to design a control strategy that keeps the traffic dynamics cooperative, which literally means \emph{free of congestion}. From safety control perspective, the \emph{safe set} is defined as the cooperative region (or a subset of the cooperative region, as the whole cooperative region might require a large number of equations to characterize).  In addition, since the model in \cite{coogan2015controlling} is a hybrid system, restriction to this type of safe sets discards a substantial amount of modes that are capturing the non-cooperative behavior. As a result, the equations governing the evolution in the safe set (cooperative region) are much simpler than the dynamics of the system in the whole state space. This issue is discussed further in the case study at the end of the paper. 


%

\begin{figure}[t]
\begin{center}
\begin{tikzpicture}[xscale=2,yscale=2]
\draw [blue, fill=red] ( -0.1 , 0.1 ) -- ( 0.1 , 0.1 ) -- ( 0 , 0 )-- cycle;
\draw [blue, fill=red] ( -0.1 , -0.1 ) -- ( 0.1 , -0.1 ) -- ( 0 , 0 )-- cycle;
\draw [blue, fill=green] ( -0.1 , -0.1 ) -- ( -0.1 , 0.1 ) -- ( 0 , 0 )-- cycle;
\draw [blue, fill=green] ( 0.1 , -0.1 ) -- ( 0.1 , 0.1 ) -- ( 0 , 0 )-- cycle;
\draw [] ( -0.1 , -0.1 ) rectangle ( 0.1 , 0.1 );
\draw [blue, fill=red] ( 0.9 , 0.1 ) -- ( 1.1 , 0.1 ) -- ( 1 , 0 )-- cycle;
\draw [blue, fill=red] ( 0.9 , -0.1 ) -- ( 1.1 , -0.1 ) -- ( 1 , 0 )-- cycle;
\draw [blue, fill=green] ( 0.9 , -0.1 ) -- ( 0.9 , 0.1 ) -- ( 1 , 0 )-- cycle;
\draw [blue, fill=green] ( 1.1 , -0.1 ) -- ( 1.1 , 0.1 ) -- ( 1 , 0 )-- cycle;
\draw [] ( 0.9 , -0.1 ) rectangle ( 1.1 , 0.1 );
\draw [blue, fill=red] ( 1.9 , 0.1 ) -- ( 2.1 , 0.1 ) -- ( 2 , 0 )-- cycle;
\draw [blue, fill=red] ( 1.9 , -0.1 ) -- ( 2.1 , -0.1 ) -- ( 2 , 0 )-- cycle;
\draw [blue, fill=green] ( 1.9 , -0.1 ) -- ( 1.9 , 0.1 ) -- ( 2 , 0 )-- cycle;
\draw [blue, fill=green] ( 2.1 , -0.1 ) -- ( 2.1 , 0.1 ) -- ( 2 , 0 )-- cycle;
\draw [] ( 1.9 , -0.1 ) rectangle ( 2.1 , 0.1 );
\draw [blue, fill=red] ( -0.1 , -0.9 ) -- ( 0.1 , -0.9 ) -- ( 0 , -1 )-- cycle;
\draw [blue, fill=red] ( -0.1 , -1.1 ) -- ( 0.1 , -1.1 ) -- ( 0 , -1 )-- cycle;
\draw [blue, fill=green] ( -0.1 , -1.1 ) -- ( -0.1 , -0.9 ) -- ( 0 , -1 )-- cycle;
\draw [blue, fill=green] ( 0.1 , -1.1 ) -- ( 0.1 , -0.9 ) -- ( 0 , -1 )-- cycle;
\draw [] ( -0.1 , -1.1 ) rectangle ( 0.1 , -0.9 );
\draw [blue, fill=red] ( 0.9 , -0.9 ) -- ( 1.1 , -0.9 ) -- ( 1 , -1 )-- cycle;
\draw [blue, fill=red] ( 0.9 , -1.1 ) -- ( 1.1 , -1.1 ) -- ( 1 , -1 )-- cycle;
\draw [blue, fill=green] ( 0.9 , -1.1 ) -- ( 0.9 , -0.9 ) -- ( 1 , -1 )-- cycle;
\draw [blue, fill=green] ( 1.1 , -1.1 ) -- ( 1.1 , -0.9 ) -- ( 1 , -1 )-- cycle;
\draw [] ( 0.9 , -1.1 ) rectangle ( 1.1 , -0.9 );
\draw [blue, fill=red] ( 1.9 , -0.9 ) -- ( 2.1 , -0.9 ) -- ( 2 , -1 )-- cycle;
\draw [blue, fill=red] ( 1.9 , -1.1 ) -- ( 2.1 , -1.1 ) -- ( 2 , -1 )-- cycle;
\draw [blue, fill=green] ( 1.9 , -1.1 ) -- ( 1.9 , -0.9 ) -- ( 2 , -1 )-- cycle;
\draw [blue, fill=green] ( 2.1 , -1.1 ) -- ( 2.1 , -0.9 ) -- ( 2 , -1 )-- cycle;
\draw [] ( 1.9 , -1.1 ) rectangle ( 2.1 , -0.9 );
\draw [line width=0.4mm, ->] ( -0.9 , 0 ) -- ( -0.1 , 0 );
\draw [line width=0.4mm,->] ( 0.1 , 0 ) -- ( 0.9 , 0 );
\draw [line width=0.4mm,->] ( 1.1 , 0 ) -- ( 1.9 , 0 );
\draw [line width=0.4mm,->,dashed] ( 2.1 , 0 ) -- ( 2.9 , 0 );
\draw [line width=0.4mm,<-,dashed] ( -0.9 , -1 ) -- ( -0.1 , -1 );
\draw [line width=0.4mm,<-] ( 0.1 , -1 ) -- ( 0.9 , -1 );
\draw [line width=0.4mm,<-] ( 1.1 , -1 ) -- ( 1.9 , -1 );
\draw [line width=0.4mm,<-] ( 2.1 , -1 ) -- ( 2.9 , -1 );
\draw [line width=0.4mm,<-,dashed] ( 0 , 0.7 ) -- ( 0 , 0.1 );
\draw [line width=0.4mm,<-] ( 0 , -0.1 ) -- ( 0 , -0.9 );
\draw [line width=0.4mm,<-] ( 2 , -0.1 ) -- ( 2 , -0.9 );
\draw [line width=0.4mm,<-] ( 0 , -1.1 ) -- ( 0 , -1.7 );
\draw [line width=0.4mm,<-,dashed] ( 2 , 0.7 ) -- ( 2 , 0.1 );
\draw [line width=0.4mm,<-] ( 2 , -1.1 ) -- ( 2 , -1.7 );
\draw [line width=0.4mm,->] ( 1 , -0.1 ) -- ( 1 , -0.9 );
\draw [line width=0.4mm,->] ( 1 , -0.1 ) -- ( 1 , -0.9 );
\draw [line width=0.4mm,->] ( 1 , 0.7 ) -- ( 1 , 0.1 );
\draw [line width=0.4mm,->,dashed] ( 1 , -1.1 ) -- ( 1 , -1.7 );
\draw (-0.5,0.1) node[] {$1$};
\draw (0.5,0.1) node[] {$2$};
\draw (1.5,0.1) node[] {$3$};
\draw (0.5,-1.1) node[] {$6$};
\draw (1.5,-1.1) node[] {$5$};
\draw (2.5,-1.1) node[] {$4$};
\draw (0.85,0.5) node[] {$9$};
\draw (1.85,-0.5) node[] {$12$};
\draw (-0.1,-0.5) node[] {$8$};
\draw (0.85,-0.5) node[] {$10$};
\draw (-0.1,-1.5) node[] {$7$};
\draw (1.85,-1.5) node[] {$11$};
\draw (0.15,-0.2) node[] {$a$};
\draw (1.15,-0.2) node[] {$b$};
\draw (2.15,-0.2) node[] {$c$};
\draw (0.15,-1.2) node[] {$d$};
\draw (1.15,-1.2) node[] {$e$};
\draw (2.15,-1.2) node[] {$f$};

\end{tikzpicture}

\caption{An urban traffic network. Each directed edge $1-12$ represents a one-way road. The vertices $a-f$ are junctions. The control input is a 6-dimensional tuple, where each component represents the decision for the traffic light at each junction from the set $\left \{NS,EW \right \}$, where $NS$ and $EW$ stand for the actuation of the vehicular flow in the north-south and the east-west directions, respectively.}
\label{fig:network}
\end{center}
\end{figure}

\subsection{Problem Formulation}
We consider discrete time systems in the form of
\begin{equation}
\label{eq:dynamics}
x^+=f(x,w,u),
\end{equation}
where $ x \in \mathbb{R}_+^n $ is the state, $w \in \mathcal{W}$ is the adversarial input and $u \in \mathcal{U}$ is the control input from an admissible set $\mathcal{U}$. We assume that the set $\mathcal{W} \subset \mathbb{R}_+^m$ is a rectangle in the form of:
\begin{equation}
\label{eq:adv}
\mathcal{W}=\mathcal{R}({w^*}),
\end{equation}
which is a reasonable assumption for many networked systems where the components of the adversarial inputs are stochastically independent. Note that any set $\mathcal{W}$ can be over-approximated by a $\mathcal{R}({w^*})$.
We do not make any restrictive assumptions on $\mathcal{U}$. For instance, $\mathcal{U}$ is an index set in an urban traffic network. 
\begin{define}
System \eqref{eq:dynamics} is \emph{cooperative} if for all $x_1 \preceq x_2 , w_1,\preceq w_2$:
\begin{equation}
\label{eq:mono}
f(x_1,w_1,u) \preceq f(x_2, w_2, u),~ \forall u \in \mathcal{U}.
\end{equation}
\end{define}
We assume that system \eqref{eq:dynamics} is cooperative. Apart from this property, we do not further restrict the function $f:\mathbb{R}_+^n \times \mathcal{W} \times \mathcal{U} \rightarrow \mathbb{R}_+^n$. In particular, we are interested in hybrid systems. For example, the urban traffic model in \cite{coogan2015controlling} is a piece-wise affine hybrid system. See Sec. \ref{sec:case} or \cite{coogan2015controlling} for further details.  
\begin{remark}
In this paper, monotonicity is defined with respect to the state and adversarial inputs, which is different from the definitions  in \cite{angeli2003monotone}, \cite{ghaemi2011safety} and \cite{meyer2015safety}. In the mentioned works  \footnote{In \cite{angeli2003monotone} only deterministic control systems are considered.}
, for all $x_1 \preceq x_2 , w_1 \preceq w_2, u_1 \preceq u_2$:
\begin{equation*}
f(x_1,w_1,u_1) \preceq f(x_2, w_2, u_2).
\end{equation*}
Such systems are also cooperative with respect to the control inputs. We have relaxed this condition in this paper. We do not even assume that the set $\mathcal{U}$ is partially ordered. 
\end{remark}

We wish to restrict the evolution of the state of the system to a user-defined set, which is referred to as  \emph{safe set} in the rest of the paper. We assume that safe sets are lower-sets. This is a restrictive assumption that is specifically motivated by the nature of the urban traffic networks and is also closely related to the stabilization of cooperative systems in the first orthant. 
The problems formulated in \cite{ghaemi2011safety} and \cite{meyer2015safety} consider a more general form of safe sets that are not necessarily lower-sets. 
In this paper, we consider the following problem: 




\begin{problem}
\label{problem}
Given a cooperative system \eqref{eq:dynamics} and a lower-set safe-set $\mathcal{S} \subset \mathbb{R}_+^n$, find a  set of initial conditions  and a control strategy such that the evolution of the system, for any sequence of admissible adversarial inputs, is confined to $\mathcal{S}$ for all times.
\end{problem}

The solution to the problem above involves computation of a set $\Omega \subseteq \mathcal{S}$ and a control policy $h: \Omega \rightarrow \mathcal{U}$, such that the evolution of the system is restricted to $\Omega$. The set $\Omega$ is a \emph{robust control invariant set} (RCIS), which is formally defined in Sec. \ref{sec:robust}. We may also find the \emph{maximal robust control invariant set} (MRCIS), which corresponds to the complete solution to  Problem \ref{problem}.  However, finding MRCIS is not always computationally practical. Instead, we focus on a more tractable solution with some possible conservativeness. 
The main drawback of conservativeness is that if we can not find a RCIS, we can not claim that the MRCIS is non-existent (empty). We investigate the limitations of our approach in Sec. \ref{sec:necessary}. Informally, we show that if our approach is not able to find a RCIS (a solution to Problem \ref{problem}), it is very likely that MRCIS is empty (there does not exist a solution to Problem \ref{problem}). 
\section{Robust Controlled Invariant Set}
\label{sec:robust}
In this section, we explain how to find a RCIS inside the safe set $\mathcal{S}$. We begin with the definition of RCIS. Next, we focus on MRCIS and explain its geometrical features and computational limitations. Then the key method of this paper is presented.

\begin{define}
\label{eq:dyn}
Given system \eqref{eq:dynamics}, 
the set $\Omega \subseteq \mathbb{R}$ is RCIS if and only if:
\begin{equation*}
\forall x \in \Omega, \exists u \in \mathcal{U} ~s.t.~ f(x,w,u)\in \Omega, \forall w\in \mathcal{W}.
\end{equation*}
\end{define}
The following statements are well known results (see, e.g., \cite{kerrigan2001}) that are stated without proof. 
\begin{proposition} 
If $\Omega_i$, $i=1,\cdots,n_\Omega$ are RCISs, then $\bigcup_i \Omega_i$ is also a RCIS. 
\end{proposition}
\begin{proposition} 
If there exist a RCIS $\Omega$, then there exist a unique MRCIS $\Omega_\infty$ such that $\Omega \subseteq \Omega_\infty$.
\end{proposition}

Implementing the MRCIS fixed point algorithm for a hybrid system is computationally intensive and is limited to very small systems subject to convex sets (see, e.g., \cite{kerrigan2001} for discussion) . Specifically, computing the robust predecessor involves set projection that is computationally challenging for complex systems. Moreover, finite termination is not guaranteed and early termination does not result in a RCIS (a solution to Problem \ref{problem}). Instead, we exploit monotonicity to introduce a new approach. The following lemma is the key idea of the paper. 

\begin{lemma}
\label{lemma:main}
If there exist $x_0 \in \mathcal{S}$ and a control sequence $u_0, u_1,u_2,\cdots,u_{N-1}$ such that
\begin{equation}
x_{k+1}=f(x_k,w^*,u_k), k=0,\cdots,N-1
\end{equation}
satisfies the following conditions:
\begin{enumerate}
\item $x_k \in \mathcal{S}$,
\item $\exists k^* s.t. ~x_{N} \in \mathcal{R}({x_{k^*}})$,
\end{enumerate}
then the set 
\begin{equation}
\Omega=\bigcup_{k=0}^{N-1} \mathcal{R}({x_k})
\end{equation}
is a RCIS inside $\mathcal{S}$. 
\end{lemma}
\begin{proof}
We show that for any point in $\Omega$, there exist a control such that for all adversarial inputs, the successor is in $\Omega$. For all $x^\prime \in \Omega, \exists p^\prime \le N-1 ~s.t.~ x^\prime \in \mathcal{R}(x_p)$. Now we apply $u_p$. Monotonicity implies $f(x^\prime,w,u_p) \preceq f(x_p,w^*,u_p)=x_{p+1}$. Therefore, $f(x^\prime,w,u_p) \in \mathcal{R}(x_{p+1})$. But we know that $\mathcal{R}(x_{p+1}) \subset \Omega$ for all $p=0,\cdots,N-1$, where $\mathcal{R}(x_N) \subseteq \mathcal{R}(x_k^*) \subset \Omega$ follows from condition (2). Therefore, $f(x^\prime,w,u_p) \in \Omega$. \end{proof}

\begin{figure}[t!]
\centering
\begin{tabular}{l r}
\begin{tikzpicture}[xscale=0.7,yscale=0.7]

\draw[->] (0,0) -- (0,5);
\draw[->] (0,0) -- (5,0);
\draw[dashed,line width = 0.3mm]  (4,2) -- (3,3) -- (2,3.4) -- (3.5,1.5);

\draw(4,2)node[] {\textbullet};
\draw (3,3) node[] {\textbullet};
\draw (2,3.4) node[] {\textbullet};
\draw  (3.5,1.5) node[] {\textbullet};

\draw (4.2,2.2) node[] {$x_0$};
\draw (3.3,3.3) node[] {$x_1$};
\draw (1.6,3.6) node[] {$x_2$};
\draw (3.2,1.2) node[] {$x_3$};

\draw[dashed,color=black!30]  ( 0.6 ,0) -- ( 0.6 ,5);
\draw[dashed,color=black!30]  (0, 0.4 ) -- (5, 0.4 );
\draw[dashed,color=black!30]  ( 1.2 ,0) -- ( 1.2 ,5);
\draw[dashed,color=black!30]  (0, 0.8 ) -- (5, 0.8 );
\draw[dashed,color=black!30]  ( 1.8 ,0) -- ( 1.8 ,5);
\draw[dashed,color=black!30]  (0, 1.2 ) -- (5, 1.2 );
\draw[dashed,color=black!30]  ( 2.4 ,0) -- ( 2.4 ,5);
\draw[dashed,color=black!30]  (0, 1.6 ) -- (5, 1.6 );
\draw[dashed,color=black!30]  ( 3.0 ,0) -- ( 3.0 ,5);
\draw[dashed,color=black!30]  (0, 2.0 ) -- (5, 2.0 );
\draw[dashed,color=black!30]  ( 3.6 ,0) -- ( 3.6 ,5);
\draw[dashed,color=black!30]  (0, 2.4 ) -- (5, 2.4 );
\draw[dashed,color=black!30]  ( 4.2 ,0) -- ( 4.2 ,5);
\draw[dashed,color=black!30]  (0, 2.8 ) -- (5, 2.8 );
\draw[dashed,color=black!30]  ( 4.8 ,0) -- ( 4.8 ,5);
\draw[dashed,color=black!30]  (0, 3.2 ) -- (5, 3.2 );
\draw[dashed,color=black!30]  (0, 3.6 ) -- (5, 3.6 );
\draw[dashed,color=black!30]  (0, 4.0 ) -- (5, 4.0 );
\draw[dashed,color=black!30]  (0, 4.4 ) -- (5, 4.4 );

\end{tikzpicture}
& 
\begin{tikzpicture}[xscale=0.7,yscale=0.7]
\draw[fill=cyan!60] (0,0) -- (4,0) -- (4,2) -- (3,2) -- (3,3) -- (2,3) -- (2,3.4) -- (0,3.4) -- cycle;

\draw(4,2)node[] {\textbullet};
\draw (3,3) node[] {\textbullet};
\draw (2,3.4) node[] {$\bullet$};
\draw  (3.5,1.5) node[] {$\circ$};

\draw[dashed,color=black!30]  ( 0.6 ,0) -- ( 0.6 ,5);
\draw[dashed,color=black!30]  (0, 0.4 ) -- (5, 0.4 );
\draw[dashed,color=black!30]  ( 1.2 ,0) -- ( 1.2 ,5);
\draw[dashed,color=black!30]  (0, 0.8 ) -- (5, 0.8 );
\draw[dashed,color=black!30]  ( 1.8 ,0) -- ( 1.8 ,5);
\draw[dashed,color=black!30]  (0, 1.2 ) -- (5, 1.2 );
\draw[dashed,color=black!30]  ( 2.4 ,0) -- ( 2.4 ,5);
\draw[dashed,color=black!30]  (0, 1.6 ) -- (5, 1.6 );
\draw[dashed,color=black!30]  ( 3.0 ,0) -- ( 3.0 ,5);
\draw[dashed,color=black!30]  (0, 2.0 ) -- (5, 2.0 );
\draw[dashed,color=black!30]  ( 3.6 ,0) -- ( 3.6 ,5);
\draw[dashed,color=black!30]  (0, 2.4 ) -- (5, 2.4 );
\draw[dashed,color=black!30]  ( 4.2 ,0) -- ( 4.2 ,5);
\draw[dashed,color=black!30]  (0, 2.8 ) -- (5, 2.8 );
\draw[dashed,color=black!30]  ( 4.8 ,0) -- ( 4.8 ,5);
\draw[dashed,color=black!30]  (0, 3.2 ) -- (5, 3.2 );
\draw[dashed,color=black!30]  (0, 3.6 ) -- (5, 3.6 );
\draw[dashed,color=black!30]  (0, 4.0 ) -- (5, 4.0 );
\draw[dashed,color=black!30]  (0, 4.4 ) -- (5, 4.4 );

\draw[->] (0,0) -- (0,5);
\draw[->] (0,0) -- (5,0);
\draw (2,1.5) node[] {$\bigcup \limits_{i=0}^2 \mathcal{R}({x_i})$};
\end{tikzpicture}
\end{tabular}
\caption{(Left) A hypothetical trajectory that satisfies the assumptions in Lemma \ref{lemma:main} since $x_3 \preceq x_0$. (Right) The union of lower-set boxes (shaded region) is a RCIS.}
\label{fig:robust}
\end{figure}
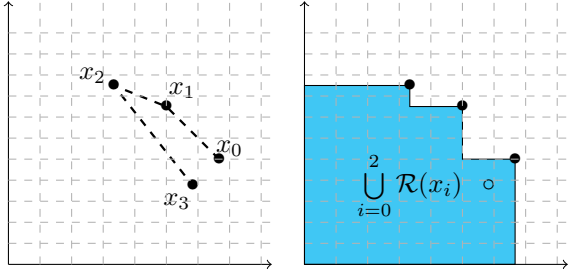

A graphical depiction of the assumptions in Lemma \ref{lemma:main} is shown in Fig. \ref{fig:robust}. Lemma \ref{lemma:main}  motivates the following definition: 

\begin{define}
\label{define:s}
An \emph{s-sequence} is a finite length sequence of controls, denoted by:
\begin{equation}
u^s:=(u_0^*,u_1^*,u_2^*,\cdots,u^*_{T-1}),
\end{equation}
where there exist $x_0^* \in \mathcal{S}$ such that 
\begin{equation}
x_T^* \preceq x_0^*,
\end{equation}
where $T$ is the length of the sequence and $x^*_{k+1}=f(x_k^*,w^*,u_k^*), x_k^* \in \mathcal{S}, 0\le k \le T-1$. 
\end{define}
The conditions in the definition above can be formulated as the set of the following constraints:
\begin{equation}
\label{eq:constraints}
\left \{
\begin{array}{l}
 x_k^* \in \mathcal{S},  0 \le k \le T-1, \\
 x_{k+1}^*=f(x_k^*,w^*,u_k^*), \\
 x_T^* \preceq x_0^*.
\end{array}
\right.
\end{equation}
The theorem below immediately follows from Lemma \ref{lemma:main}. 
\begin{theorem}
If $(x_k^*,u_k^*), 0 \le k \le T-1$, is a feasible solution to the set of constraints \eqref{eq:constraints}, then $u^s=(u_0^*,u_1^*,u_2^*,\cdots,u^*_{T-1})$ is an s-sequence and the set
\begin{equation}
\label{eq:rc}
\Omega^*:=\bigcup_{k=0}^{T-1} \mathcal{R}(x_k^*)
\end{equation}
is a RCIS inside $\mathcal{S}$.
\end{theorem}

We now explain how to use the theorem above and find an $s-sequence$.
If $T$ is fixed, finding a solution for \eqref{eq:constraints} is a feasibility problem. One way to approach this problem is formulating \eqref{eq:constraints} as an  SMT (satisfiability modulo theories) problem. There exist powerful SMT solvers that are able to handle nonlinearities in the constraints \cite{gao2013dreal}. An alternative approach is formulating \eqref{eq:constraints} as the constraints of an optimization problem, where the cost function aims to maximize a notion of \emph{size} for the set $\Omega^*$. For instance, the following optimization problem:
\begin{equation}
\label{eq:opt}
\begin{array}{lll}
u_k^*,x_k^*= & argmax  & \left \| x_0^* \right \|_1,
\\
& s.t. & \text{Eq.} ~\eqref{eq:constraints},
\end{array}
\end{equation}
provides a feasible solution to \eqref{eq:constraints} where $L_1$ norm of $x_0^*$ is maximized. As opposed to the iterative procedure in \cite{kerrigan2001}, we are able to find a RCIS for system \eqref{eq:dynamics} by solving a single optimization problem.

The dynamics of a large class of systems can be written as mixed integer constraints. In particular, piecewise affine hybrid systems and safe sets that are unions of polyhedra (not necessarily convex) can be encoded using  mixed integer linear constraints (see, e.g., \cite{bemporad1999control}). Therefore, the optimization problem above can be written as a mixed integer linear programming (MILP) problem, which is solved using efficient state of the art solvers. 
If \eqref{eq:dynamics} is a linear system and $\mathcal{S}$ is a polyhedron, then \eqref{eq:opt} is solved in polynomial time. Otherwise, the time required for solving \eqref{eq:opt} grows polynomially with respect to the size of system \eqref{eq:dynamics} and exponentially with respect to $T$ and the number of integer constraints (e.g., the number of modes of the hybrid system).

If the set of constraints \eqref{eq:constraints} is infeasible, one has to change $T$ to search for feasibility. Algorithmically, we start from $T=1$ and implement $T \gets T+1$ until \eqref{eq:constraints} becomes feasible and a solution to Problem \ref{problem} is obtained. Large values of $T$ makes finding a feasible solution for \eqref{eq:constraints} impractical.  In Sec. \ref{sec:necessary}, we establish a relation for the necessity of the existence of s-sequences.

\begin{remark}
 As mentioned earlier, for any feasible solution, we may use \eqref{eq:rc} to find a RCIS. If multiple feasible solutions are available, we may find the union of all the RCISs provided by \eqref{eq:rc} to find a larger RCIS. Practically, RCIS are useful as terminal constraints of model predictive controllers (see \cite{kerrigan2001}). Therefore, larger RCISs might be desirable. We do not yet have a proof that by taking the union of all RCISs, in the limit $T\rightarrow \infty$, we are able to get arbitrarily close to the MRCIS.
 \end{remark}

\section{Controlled Limit Cycles and Attractive Sets}
\label{sec:cycle}
In the last section, we provided a solution to Problem \ref{problem}: $\Omega^*$ is the set of initial conditions and the control strategy is based on s-sequences. In this section, we characterize the infinite time system response to the repetitions of an s-sequence and show its relation to controlled limit cycles and attractive sets.  
\begin{lemma}
\label{lemma:infinite}
Let $u^s=(u_0^*,\cdots,u_k^*)$ be the s-sequence that corresponds to $x_0^* \in \mathcal{S}$. Then the trajectory of the following system:
\begin{equation}
\label{eq:star}
x^*_{cT+k+1}=f(x^*_{cT+k},w^*,u_k^*), c=0,1,\cdots, 0\le k \le T-1,
\end{equation}
converges to a limit cycle, i.e.  $\lim_{c\rightarrow \infty} x^*_{cT+k}$ exists. 
\end{lemma}
\begin{proof}
It follows from the definition of s-sequences that $x_T^* \preceq x_0^*$. Monotonicity implies:
\begin{equation}
\begin{array}{c}
x_{T+1}^*=f(x_T^*,w^*,u_0^*) \preceq f(x_0^*,w^*,u_0^*) =x_1^*,
\\
\vdots \\
x_{2T}^*=f(x_{2T-1}^*,w^*,u_{T-1}^*) \preceq f(x_{T-1}^*,w^*,u_{T-1}^*)=x_T^*.
\end{array}
\end{equation}
By continuing the argument above we draw the conclusion that:
\begin{equation}
x_{(c+1)T+k}^* \preceq x_{cT+k}^*, c=0,1,\cdots.
\end{equation}
Therefore, each vector component of the following sequence is non-increasing:
\begin{equation}
x_k^*,x_{T+k}^*,x_{2T+k}^*,\cdots,x_{cT+k}^*,
\end{equation}
and it is already known that is lower bounded (by the origin). As a result, it follows from the \emph{cooperative convergence theorem} \cite{yeh2006real} that the limit $c \rightarrow \infty$ exists. We denote:
\begin{equation}
x_k^{\infty}:=\lim_{c\rightarrow \infty} x^*_{cT+k}.
\end{equation}
As a result, $f(x_{T-1}^\infty,w^*,u_{T-1}^*)=x_0^\infty$ and the trajectory of \eqref{eq:star} converges to $\overline{x_0^\infty, x_1^\infty, \cdots,x_{T-1}^\infty}$.
\end{proof}

We introduce the following repetitive sequence:
\begin{equation}
\label{eq:series}
\overline{u}^s:=\overline{(u_0^*,u_1^*,\cdots,u_{T-1}^*)}.
\end{equation}
The sequence above is basically the control strategy. Its applicability solely requires the initial condition to be in $\mathcal{R}(x_0^*)$ 
(it is straightforward to see from the proof of Lemma \ref{lemma:main} that $\mathcal{R}(x_0^*)$ is reachable from any point in $\Omega^*$). In other words, our solution to the control strategy in Problem \ref{problem} is unexpectedly a simple policy that does not require state feedback.

\begin{theorem}
If $x_k^*,u_k^*, 0 \le k \le T-1$, is a feasible solution to \eqref{eq:constraints}, then the set 
\begin{equation}
\label{eq:attractive}
\Gamma=\bigcup_{k=0}^{T-1} \mathcal{R}(x_k^\infty),
\end{equation}
is an \emph{attractive set} for all the trajectories of system \eqref{eq:dynamics}  starting from $\mathcal{R}(x_0^*)$ under the control strategy \eqref{eq:series}.   
\end{theorem}

\begin{proof}(sketch)
Let $x^*_0,x^*_1,x^*_2,\cdots$ and $x_0,x_1,x_2,\cdots$,
represent the trajectories of $x^*_{k+1}=f(x_k^*,w^*,u_k^*)$ and $x_{k+1}=f(x_k,w,u_k^*)$, respectively. 
Monotonicity indicates that:
 \begin{equation*}
x_{cT+k} \preceq x^*_{cT+k}, c=0,1,\cdots, 0\le k \le T-1.
\end{equation*}
As $c \rightarrow \infty$, the right hand side approaches $\Gamma$. 
Therefore, all the left hand side values also finally reach $\Gamma$ and remain there forever.  
\end{proof}

\section{Necessity of existence of s-sequences}
\label{sec:necessary}

In the last sections, we showed that the existence of s-sequences is sufficient for providing a solution to Problem \ref{problem}. In this section we provide a fundamental result on the necessity conditions for the existence of s-sequences. We show that, under some assumptions, the existence of s-sequences is \emph{almost} necessary.  
\begin{assumption}
\label{assume:safe}
The safe set $\mathcal{S}$ is bounded. 
\end{assumption}


\begin{assumption}(Strict monotonicity with respect to the adversarial inputs)
There exist $\alpha>0$ such that for all $x \in \mathbb{R}_+^n, u \in \mathcal{U}$ and $w_1,w_2$ such that:
\begin{equation}
w_1 + \varepsilon 1_n  \preceq w_2,
\end{equation}
where $1_n$ is a n-dimensional vector of all ones and $\varepsilon>0$, the following relation holds:
\begin{equation}
f(x,u,w_1) + \alpha \varepsilon 1_n  \preceq f(x,u,w_2).
\end{equation}
\label{assume:strict}
\end{assumption}
We now use the assumptions above to provide the key idea of this section.

\begin{lemma}
\label{lemma:pigeon}
If there exist a robust safety control strategy $u=h(x), h:\Omega \rightarrow \mathcal{U}$, such that the trajectory of system \eqref{eq:dynamics} with $\mathcal{W}=\mathcal{R}(w^*)$ is restricted to $\mathcal{S}$, then there exist at least one s-sequence with length $T$ for system \eqref{eq:dynamics} with $\mathcal{W}=\mathcal{R}(w^*-1_n \varepsilon)$ such that 
\begin{equation}
T \le \frac{c}{(\alpha \varepsilon)^n},
\label{eq:lemmabound}
\end{equation}
where $c$ is a constant solely depending on $\mathcal{S}$, $0<\varepsilon < w^*$, and $\alpha$ is defined in Assumption \ref{assume:strict}. 
\end{lemma}
\begin{proof} (sketch)
Consider a uniform grid over the set $\mathcal{S}$ with cube cells of length $\varepsilon$. The number of cells $\mathcal{N}$ is  proportional to $\frac{1}{\varepsilon^n}$, so we let $\mathcal{N}=\frac{c}{\varepsilon^n}$, where $c$ depends on the shape of $\mathcal{S}$.  
Now consider a safe trajectory for system $ x_{k+1}=f(x_k,w_k^*,u_k)$ such that the trajectory does not meet the conditions in Lemma \ref{lemma:main}. By the virtue of the \emph{pigeonhole principle}, after $\mathcal{N}+1$ points obtained from the trajectory, there exist a cell that contains at least two points. In other words, without loss of generality, by redefining $x_0$ as the earlier point in the cell, there exist $T\le \mathcal{N}$ such that
\begin{equation}
x_T-x_0 \preceq \epsilon 1_n.
\label{equ:compare1}
\end{equation} 
If the same control sequence, $u_0,u_1,\cdots,u_{T-1}$, is applied to the system $x_{k+1}^\prime=f(x_k^\prime,w_k^*-1_n \varepsilon,u_k)$, $x_0^\prime=x_0$, it follows from Assumption \ref{assume:strict} that
\begin{equation}
x_T^\prime + \alpha \varepsilon 1_n \preceq x_T.
\label{equ:compare2}
\end{equation}
By comparing \eqref{equ:compare1} and \eqref{equ:compare2}, we obtain that $x_T^\prime \preceq x_0$, which indicates  that $(u_0,u_1,\cdots,u_{T-1})$ is an s-sequence for system \eqref{eq:dynamics} where $\mathcal{W}=\mathcal{R}(w^*-1_n \varepsilon)$ and the following bound is obtained: 
$
T \le \frac{c}{(\alpha \varepsilon)^n}.
$
\end{proof}

\begin{theorem}
\label{theorem:necessary}
Provided that Assumption \ref{assume:safe} and Assumption \ref{assume:strict} are true, the existence of an s-sequence is \emph{almost} necessary for the existence of a solution to Problem \ref{problem} in the sense that:
\begin{enumerate}
\item if a robust safe control strategy for system \eqref{eq:dynamics} with $\mathcal{W}=\mathcal{R}(w^*)$ \emph{exists}, then there exist at least one s-sequence of length less than $T$ for the system \eqref{eq:dynamics} with $\mathcal{W}=\mathcal{R}(w^*-1_n \varepsilon)$ such that $T \le \frac{c_1}{\varepsilon^n}$, 
\item if an s-sequence of length less than $T$ is not found for the system \eqref{eq:dynamics} with $\mathcal{R}(w^*)$, then there does not exist a robust safe control strategy for the system \eqref{eq:dynamics} with $\mathcal{W}=\mathcal{R}(w^*+1_n \varepsilon)$ such that $\varepsilon \ge \frac{c_2}{T^\frac{1}{n}}$,
\end{enumerate}
where $c_1$ and $c_2$ are $\varepsilon$ independent constants.  
\end{theorem}
The theorem above addresses the concern of searching for very long s-sequences. Starting from $T=1$ and ending at some $T$ that is beyond our computational resources, without having an s-sequence found, we know that the existence of a solution to Problem \ref{problem} is highly unlikely. Informally, such a policy, if exists, is \emph{fragile}, in the sense that, a slight increase in the adversarial inputs makes the policy invalid. 

We conclude this section by mentioning that the results of this section are still theoretical and preliminary. We did not explain how to determine  $\alpha$ for a cooperative system. Furthermore, the approach based on the number of cells in a uniform grid may lead to very wide bounds  in Theorem \ref{theorem:necessary} that seem conservative for practical use.

\section{Case studies}
\label{sec:case}
In this section, we provide two case studies. The first case study is an academic example in two dimensions hence it is convenient to graphically illustrate the results. The second case study is of practical interest, where we apply our methods to the urban traffic network shown in Fig. \ref{fig:network}. 

\subsection{Case Study 1: Two-mode planar hybrid system}
Consider \eqref{eq:dynamics} to be the following system in $\mathbb{R}_+^2$: 
\begin{equation*}
f(x,w,u)= \left \{
\begin{array}{ll}
A_1x+w, & u=1, \\
A_2x+w, & u=2, 
\end{array}
\right.
\end{equation*}
where $x=(x_1,x_2)^T$, $w \in \mathcal{R}(w^*)$, $w^*=(0.2,0.1)^T$, and 
\begin{equation*}
A_1=\left (
\begin{array}{cc}
1.5 & 0.1 \\
0.2 & 0.5 
\end{array}
\right),
A_2=\left (
\begin{array}{cc}
0.7 & 0.1 \\
0.1 & 1.1 
\end{array}
\right).
\end{equation*}
The system above represents a two-mode hybrid (switched) system with additive disturbances where the control input set is $\mathcal{U}=\left\{ 1,2 \right\}$. Note that if $u$ is fixed, trajectories grow unbounded.  
We wish to find a control policy that restricts the evolution of the system to the safe set 
$$\mathcal{S}=\left \{ x \big| x_1+x_2 \le 50 \right \},$$ 
which is a triangular lower-set.
We encode the system above as the set of the following mixed-integer constraints:
\begin{equation*}
\left \{
\begin{array}{l}
A_1 x_k^*+w^* - M (u^*_k-1) (1 ~ 1)^T \preceq x^*_{k+1}, \\
x^*_{k+1} \preceq A_1 x_k^*+w^* + M (u^*_k-1) (1 ~ 1)^T, \\
A_2 x_k^*+w^* - M (2-u^*_k) (1 ~ 1)^T \preceq x^*_{k+1}, \\
x^*_{k+1} \preceq A_2 x_k^*+w^* + M (2-u^*_k) (1 ~ 1)^T,
\end{array}
\right.
\end{equation*}
where $M$ is a sufficiently large number ($1000$ in our implementation). We setup the optimization problem \eqref{eq:opt} as a MILP.
\subsection*{Results}
Using the Gurobi MILP solver \cite{optimization2014inc}, we find that the smallest $T$ that renders the MILP feasible is $T=7$. The solution is found almost instantly on a personal computer. The following s-sequence is obtained:
\begin{equation*}
u^s=(1,2,2,1,2,2,2),
\end{equation*}
which corresponds to $x_0^*=(16.15,33.85)^T$, $x_7^*=(16.15,33.21)^T$. We find the RCIS $\Omega$ using \eqref{eq:rc}. As explained in Sec. \ref{sec:cycle}, by applying the control sequence $\overline{(1,2,2,1,2,2,2)}$ to $x_{k+1}^*=f(x_k^*,w^*,u^*_k)$, we arrive at the limit cycle $\overline{x_0^\infty, \cdots,x_6^\infty,x_0^\infty}$, where $x_0^\infty=(13.62,27.78)^T$. The attractive set $\Gamma$ is found using \eqref{eq:attractive}. We also simulate a trajectory of system $x_{k+1}=f(x_k,w,u_k^*)$. The values of $w$ are drawn from a uniform distribution over $\mathcal{R}(w^*)$.
The results are illustrated in Fig. \ref{fig:case}. 


\begin{center}
\begin{figure*}[t]
\begin{tabular}{ccc}
\includegraphics[width=0.3\textwidth]{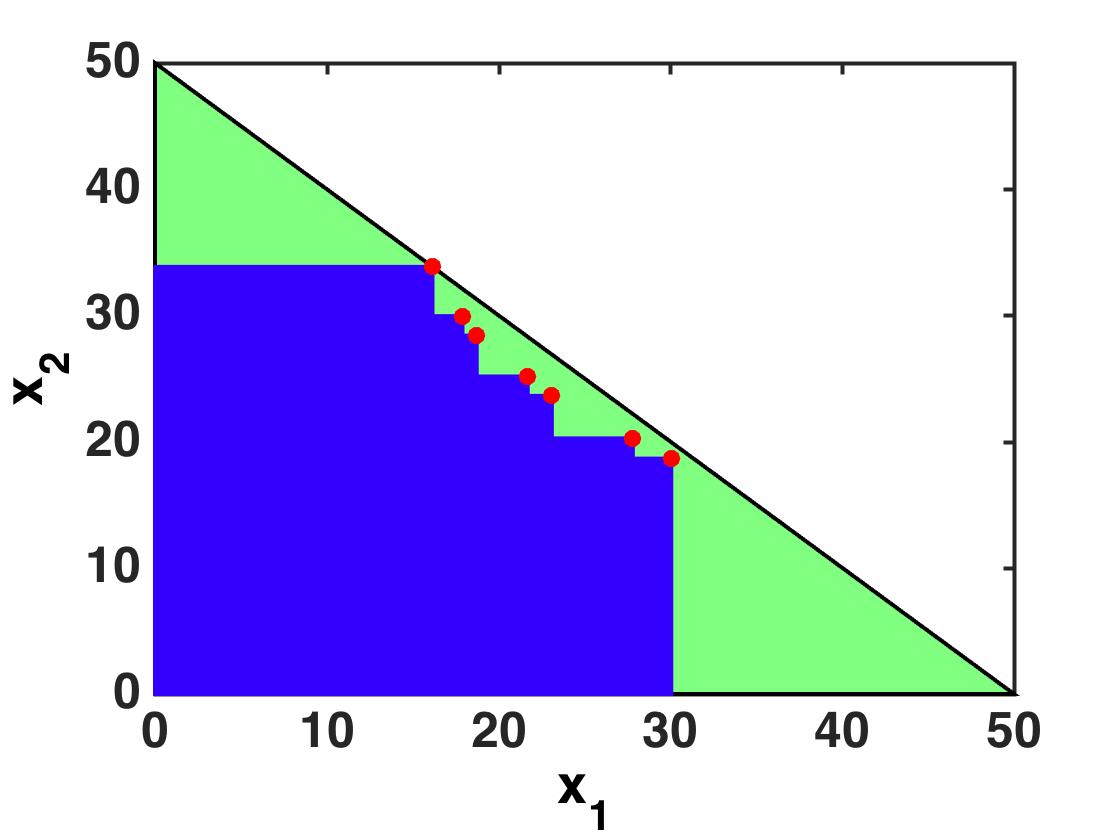} &  \includegraphics[width=0.3\textwidth]{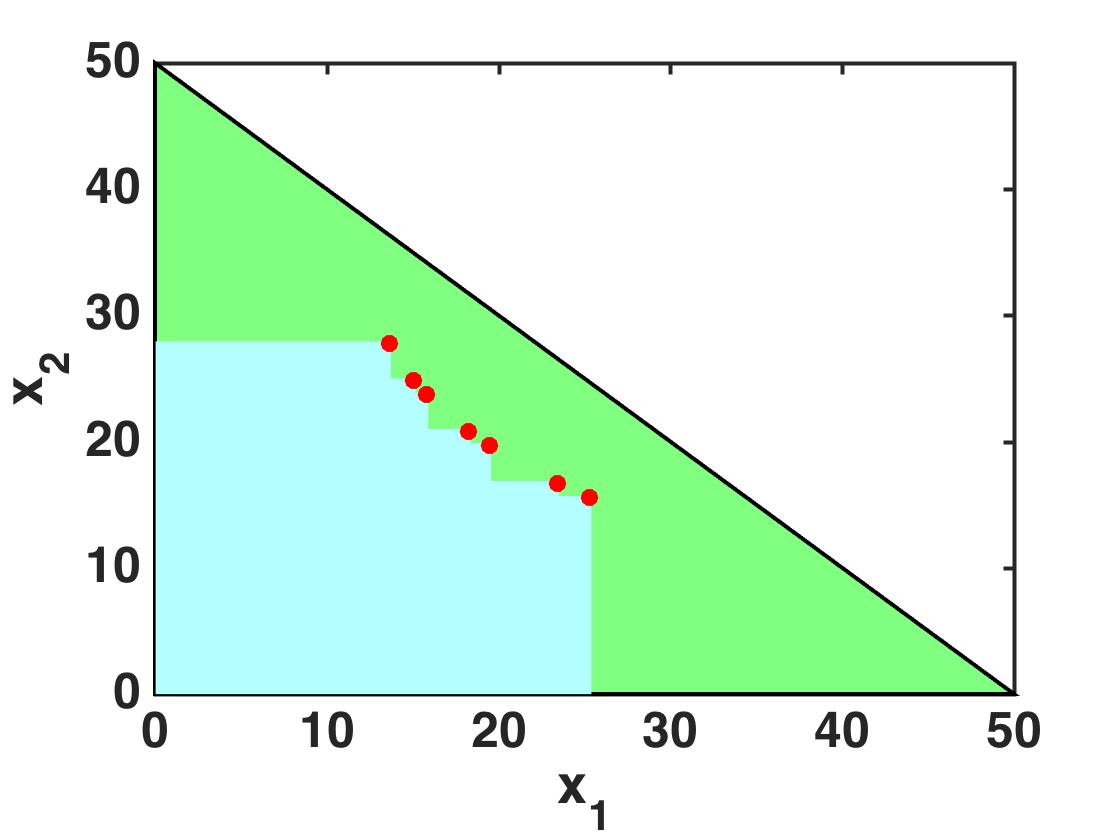} &\includegraphics[width=0.3\textwidth]{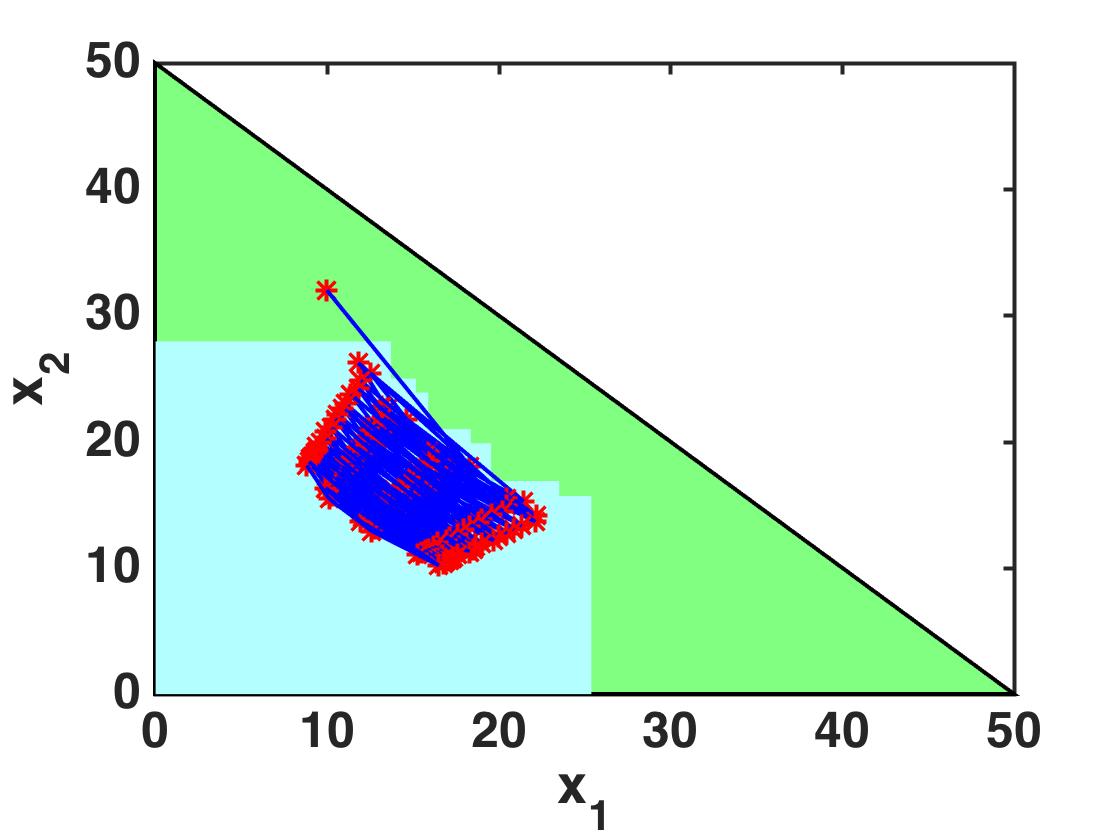} 
\end{tabular}
\caption{ Case Study 1: (Left) The blue region is RCIS $\Omega^*$ inside the green region $S$. The red points at the corners of the boxes are $x_k^*,0\le k\le 6$. (Middle) The cyan region is the attractive set $\Gamma$. The corner red points are $x_k^\infty,0\le k\le 6$. (Right) The trajectory of system \eqref{eq:dynamics} starting from $x_0=(10,32)$ under the control strategy \eqref{eq:series}. It can be seen that the trajectory reaches $\Gamma$ and stays there forever.}
\label{fig:case}
\end{figure*}
\end{center}

\subsection{Case study 2: Urban traffic network}
\label{sec:case2}

First, we explain the details of the model in \cite{coogan2015controlling}. Let $\mathcal{L}$ and $\mathcal{J}$ represent the set of links and junctions, respectively. Link $l$ is characterized by its \emph{tail junction} $\tau(l) \in \left\{\mathcal{J} \cup \emptyset\right\}$ and \emph{head junction} $\eta(l) \in \mathcal{I}$, where $\tau(l)=\emptyset$ indicates that link $l$ is an entry link to the network. We say that link $k$ is a \emph{downstream }link for $l$ if $\eta(l)=\tau(k)$. Similarly, link $l$ is an \emph{upstream} link for $k$. 
For simplicity, we consider networks in which all links are either in north-south ($NS$) or east-west ($EW$) directions. We denote the direction of link $l$ by $dir(l) \in \left\{NS,EW\right\}$. The traffic light at junction $j \in \mathcal{J}$ is denoted by $u^{(j)} \in \{NS,EW\}$. The control input is a $\left |\mathcal{J}\right |$ dimensional tuple representing all the traffic lights in the network. 
The state is $x \in \mathbb{R}_+^n$, where $n=\left |\mathcal{L} \right|$ and $x^{(l)}$ is the number of vehicles on link $l$. The number of vehicles that flow out of link $l$ in one time step, denoted by $z^{(l)}$, is:
\begin{equation}
\label{eq:flow}
z^{(l)}= \left \{
\begin{array}{cl}
\min \left(x^{(l)},c^{(l)},\underset{k, \eta(l)=\tau(k)} \min s_{lk} \right), & u^{(\eta(l))}=dir(l), \\
0, & \text{otherwise},
\end{array}
\right.
\end{equation}
where $c^{(l)}$ is the maximum outflow of vehicles from $l$ in one time step and $s_{lk}$ is the supply available from downstream link $k$ to $l$. The FIFO-based model for supply is $s_{lk}=\frac{\alpha_{lk}}{\beta_{lk}} (x^{(k),cap}-x^{(k)})$, where $\alpha_{lk} \in [0,1]$ is the capacity ratio of $k$ dedicated to $l$, $\beta_{lk} \in [0,1]$ is the ratio of flow turning from $l$ to $k$ and $x^{(k),cap} \in \mathbb{R}_+^n$ is the vehicular capacity of link $k$. As mentioned in Sec. \ref{sec:problem}, monotonicity does not hold when supply limits the flow at diverging junctions. Therefore, by restricting the state to the following rectangular safe set:
\begin{equation}
\mathcal{S}=\left\{ x \big |  x^{(l)} \le x^{(l),s} \right\}, 
\end{equation}
where $x^{(l),s} \le x^{(l),cap}-\underset{k, \eta(k)=\tau(l)} \max  \frac{\alpha_{lk}}{\beta_{lk}} c^{(k)}$, we ensure that $s_{lk}$ is never the minimizer in \eqref{eq:flow}.
 As a result, \eqref{eq:flow} becomes:
\begin{equation}
\label{eq:cooperative}
z^{(l)}= \left \{
\begin{array}{cl}
\min \left(x^{(l)},c^{(l)}\right), & u^{(\eta(l))}=dir(l), \\
0, & \text{otherwise}. 
\end{array}
\right.
\end{equation} 
The discrete time evolution of $x^{(l)}$ is given by:
\begin{equation}
x^{(l),+}=x^{(l)}-z^{(l)}+ w^{(l)}+ \underset{k, \eta(k)=\tau(l)} \sum \beta_{kl} z^{(k)},
\end{equation}
where $w^{(l)} \in [0,w^{(l),*}]$ is the adversarial input corresponding to link $l$. It is straightforward to check that $\frac{\partial x^{(l),+}}{\partial x^{(l)}} \in \{0,1\}$, $\frac{\partial x^{(l),+}}{\partial x^{(k)}} \in \{0,\beta_{kl}\}$, $\frac{\partial x^{(l),+}}{\partial w^{(l)}} =1$ and $\frac{\partial x^{(l),+}}{\partial w^{(k)}} =0$. Therefore, the evolution of each state component is cooperative with respect to the state and adversarial inputs. Finally, in a compact form, the evolution  can be written in the form \eqref{eq:dynamics}. 
We wish to find a control policy for the urban traffic network shown in Fig. \ref{fig:network} such that the state  is always in $\mathcal{S}$. The network parameters are given in Table I. 

\begin{table}[t]
\centering
\label{table:values}
\caption{Parameters of the network in Fig. \ref{fig:network}}
\begin{tabular}[c]{ |c | }
\hline
$x^{(l),s}=60$, $l=1,2,3,4,5,6,9,10$, \\
$x^{(l),s}=60$, $l=7,8,11,12$
\\ 
\hline
$c^{(l),s}=20$, $l=1,2,3,4,5,6,9,10$, \\
$c^{(l),s}=10$, $l=7,8,11,12$
\\
\hline
$\beta_{12}=0.7,\beta_{45}=\beta_{78}=\beta_{9~10}=0.7$, $\beta_{23}=\beta_{56}=0.6$,
\\
$\beta_{11~5}=\beta_{11~12}=0.5$,$\beta_{82}=\beta_{2~10}=0.4$
\\ 
$\beta_{93}=\beta_{10~6}=\beta_{11~5}=\beta_{68}=\beta_{4~12}=0.3$
\\
\hline
$w^{(1),*}=w^{(4),*}=8$, $w^{(7),*}=4$, $w^{(9),*}=7$, $w^{(11),*}=6$ \\
$w^{(l),*}=0$, $l=2,3,5,6,8,10,12$
\\
\hline
\end{tabular}
\end{table}

\subsection*{Results}


We formulate \eqref{eq:opt} as a MILP. The smallest $T$ for which an s-sequence is found is $T=5$. The time required to solve the MILP using Gurobi is 79 seconds on a 3GHz Core i7 MacBook Pro. In comparison to finite state-based safety game implemented in\cite{sadraddini2016provably}, a problem of this size (12 links, 6 junctions) is intractable, unless a very coarse partitioning of the state space is considered. 

Table II shows the traffic light at each junction for each control input in $(u_0^*,u_1^*,u_2^*,u_3^*,u_4^*)$. We also find that:
\begin{equation*}
x_0^*=(48,14,54,48,17.66,54,4,12.47,28,60,28,29)^T.
\end{equation*}
We obtain a RCIS and an attractive set that lie in $\mathbb{R}_+^{12}$. As explained in Sec. \ref{sec:necessary}, we can simulate the system $x_{k+1}^*=f(x^*_k,w^*,u_k^*)$ to obtain the limit cycle, which is illustrated in Fig. \ref{fig:traf1}. A trajectory of the system starting from $x_0^*$ with $w$ chosen from a uniform distribution over $\mathcal{R}(w^*)$ is also shown in Fig. \ref{fig:traf2}. Note that all the components of the trajectory in Fig. \ref{fig:traf2} are upper bounded by their corresponding values in the trajectory in Fig. \ref{fig:traf1}.

\centering
\begin{table}[t]
\centering
\caption{Traffic lights at junctions corresponding to the s-sequence}
\begin{tabular}{| c | c | c | c | c | c |}
\hline
junction & $u_0^*$ & $u_1^*$ & $u_2^*$ & $u_3^*$ & $u_4*$  \\ 
\hline
 $a$ & $NS$ & $EW$ & $NS$ & $NS$ & $EW$  \\ 
\hline
 $b$ & $NS$ & $NS$ & $EW$ & $EW$ & $EW$  \\ 
\hline 
 $c$ & $NS$ & $EW$ & $NS$ & $EW$ & $NS$  \\
\hline
 $d$ & $NS$ & $EW$ & $NS$ & $NS$ & $EW$  \\
\hline
 $e$ & $NS$ & $NS$ & $EW$ & $EW$ & $EW$  \\
\hline
 $f$ & $NS$ & $EW$ & $NS$ & $EW$ & $NS$  \\
\hline
\end{tabular}
\end{table}

\begin{figure*}[t]
\centering
\includegraphics[width=0.97\textwidth]{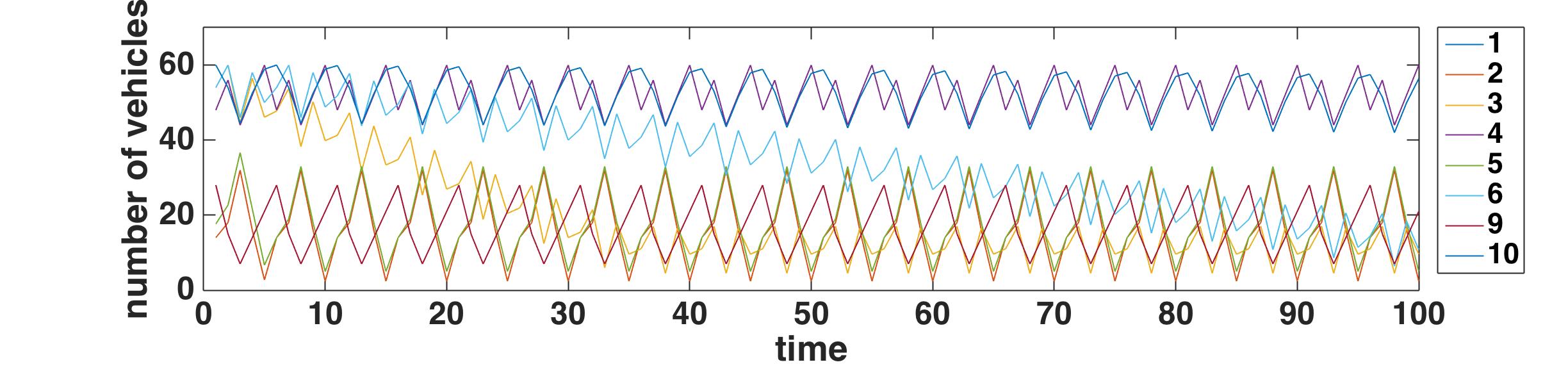} \\  \includegraphics[width=0.97\textwidth]{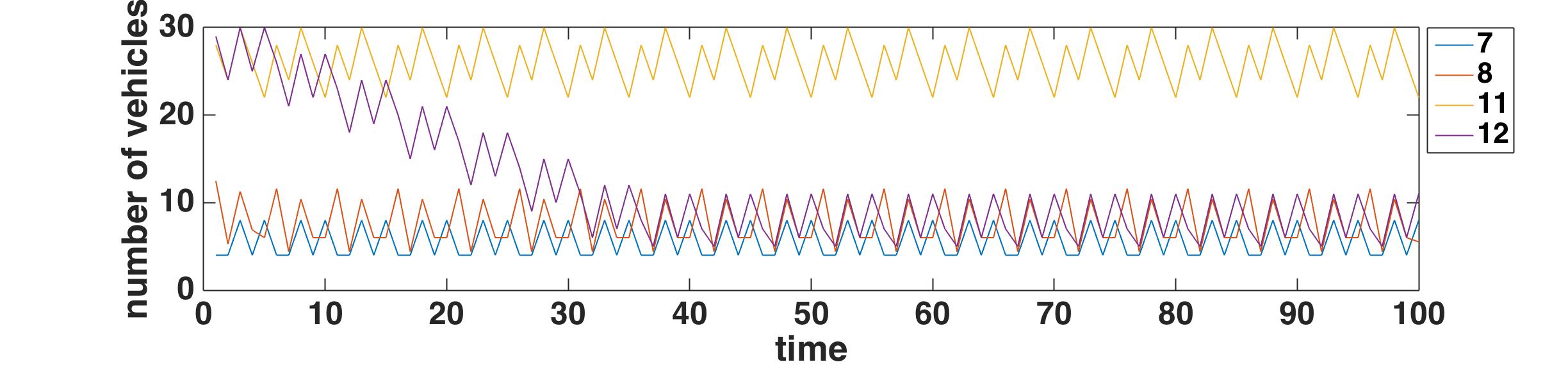}
\caption{Case Study 2: The trajectory of $x_{k+1}^*=f(x^*_k,w^*,u_k^*)$ converges to a limit cycle.}
\label{fig:traf1}
\end{figure*}

\begin{figure*}[t]
\centering
\includegraphics[width=0.97\textwidth]{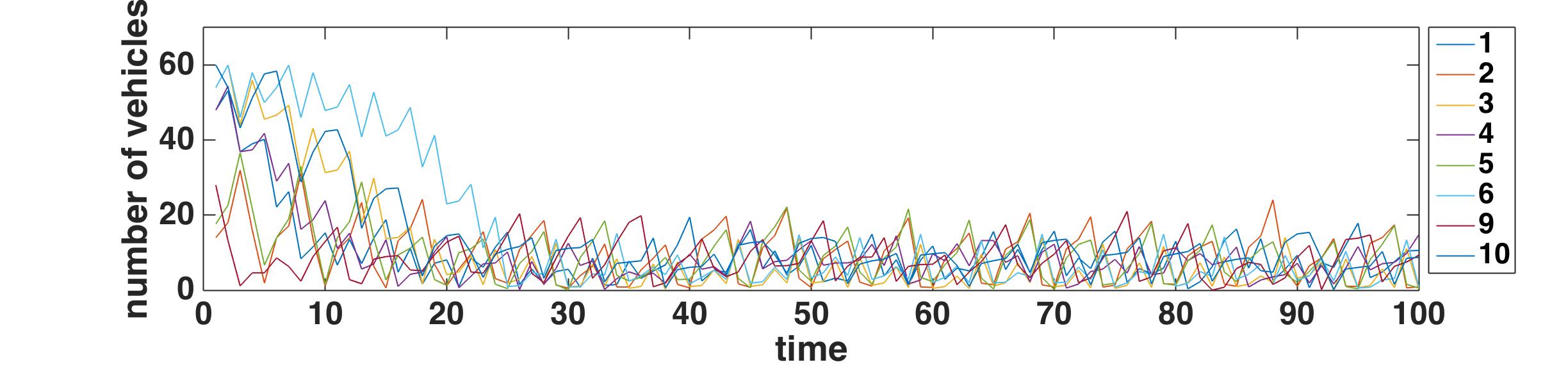} \\  \includegraphics[width=0.97\textwidth]{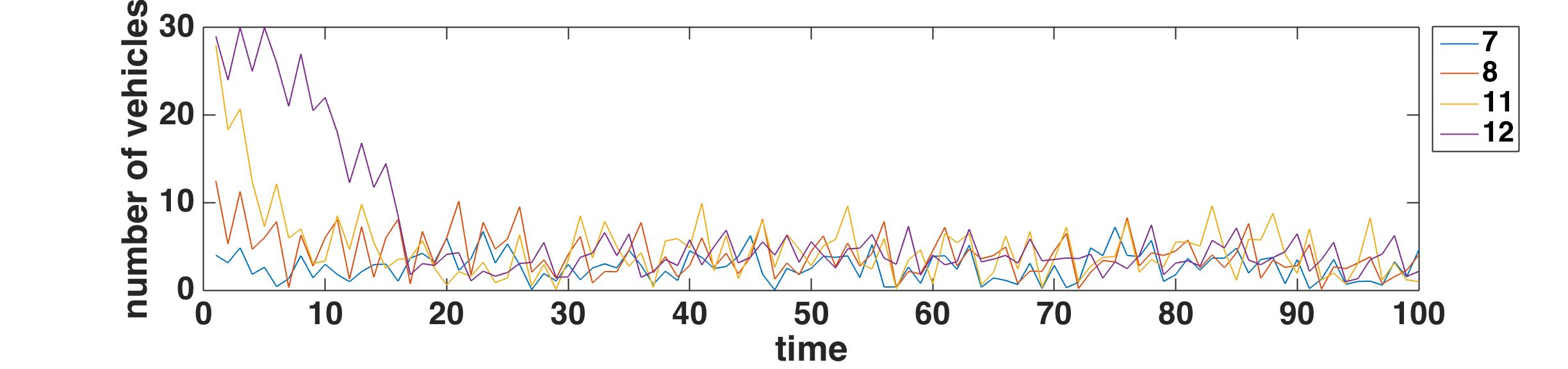}
\caption{Case Study 2: A trajectory of the system $x_{k+1}=f(x_k,w,u_k^*)$ always remains in the safe set.}
\label{fig:traf2}
\end{figure*}

\bibliography{bib_monotone}

\end{document}